\documentclass[11pt]{article}
\usepackage{amsmath,amsthm,amssymb,fullpage,bm,color,xspace}
\usepackage[colorlinks=true]{hyperref}

\usepackage{algorithm}
\usepackage[noend]{algpseudocode}
\usepackage{algorithmicx}

\newtheorem{theorem}{Theorem}[section]
\newtheorem{lemma}[theorem]{Lemma}
\newtheorem{corollary}[theorem]{Corollary}
\newtheorem{proposition}[theorem]{Proposition}

\theoremstyle{plain}
\newtheorem{definition}[theorem]{Definition}

\newcommand{\bbR}{\mathbb{R}}
\newcommand{\bbN}{\mathbb{N}}
\newcommand{\bbZ}{\mathbb{Z}}

\newcommand{\bv}{{\bm v}}

\newcommand{\bx}{{\bm x}}
\newcommand{\by}{{\bm y}}
\newcommand{\bz}{{\bm z}}
\newcommand{\bb}{{\bm b}}
\newcommand{\bzero}{{\bm 0}}
\newcommand{\bone}{{\bm 1}}
\newcommand{\btheta}{{\bm \theta}}

\newcommand{\caS}{\mathcal{S}}

\newcommand{\set}[1]{\{#1\}}
\newcommand{\E}{\mathop{\mathbf{E}}}
\newcommand{\poly}{\mathrm{poly}}
\newcommand{\polylog}{\mathrm{polylog}}
\newcommand{\compsmall}{n^4+n^2\log (n/\delta )/\epsilon^2}
\newcommand{\compgeneral}{\frac{nm^2}{\epsilon^3}\log \frac{n m}{\epsilon \delta}+\frac{n^4}{\epsilon}+\frac{n^2}{\epsilon^3}\log \frac{1}{\epsilon \delta}}
\newcommand{\compgeneralM}{\frac{nM^2}{\epsilon^3}\log \frac{n M}{\epsilon}+\frac{n^4}{\epsilon}+\frac{n^2}{\epsilon^3}\log \frac{1}{\epsilon}}

\newcommand{\compgl}{O\Bigl(\frac{n^4\polylog(n)}{(1-c_g)^2}\Bigr) \cdot  \Bigl(\frac{1}{\epsilon}\log \frac{1}{1-c_g} \Bigr)^{\poly(1/\epsilon)/(1-c_g)}}
\newcommand{\compf}{O\Bigl(\frac{n^4\polylog(n)}{(1-c_f)^2}\Bigr) \cdot  \Bigl(\frac{1}{\epsilon}\log \frac{1}{1-c_f} \Bigr)^{\poly(1/\epsilon)/(1-c_f)}}

\newcommand{\compfsimple}{O\Bigl(n^4\polylog(n) \cdot  \bigl(\frac{1}{\epsilon} \bigr)^{\poly(1/\epsilon)}\Bigr)}

\newcommand{\compintro}{O\Bigl(n^5+n^4\polylog(n) \cdot  \bigl(\frac{1}{\epsilon} \bigr)^{\poly(1/\epsilon)}\Bigr)}

\newcommand{\Vondrak}{Vondr{\'a}k\xspace}

\begin{document}

\title{Maximizing a Monotone Submodular Function with a Bounded Curvature under a Knapsack Constraint}

\author{
  Yuichi Yoshida\footnote{Supported by JSPS Grant-in-Aid for Young Scientists (B) (26730009), MEXT Grant-in-Aid for Scientific Research on Innovative Areas (24106003), and JST, ERATO, Kawarabayashi Large Graph Project.}\\
  National Institute of Informatics \emph{and} Preferred Infrastructure, Inc.\\
  \texttt{yyoshida@nii.ac.jp}
}

\maketitle
\begin{abstract}
  We consider the problem of maximizing a monotone submodular function under a knapsack constraint.
  We show that, for any fixed $\epsilon > 0$, there exists a polynomial-time algorithm with an approximation ratio $1-c/e-\epsilon$, where $c \in [0,1]$ is the (total) curvature of the input function.
  This approximation ratio is tight up to $\epsilon$ for any $c \in [0,1]$.
  To the best of our knowledge, this is the first result for a knapsack constraint that incorporates the curvature to obtain an approximation ratio better than $1-1/e$, which is tight for general submodular functions.

  As an application of our result,
  we present a polynomial-time algorithm for the budget allocation problem with an improved approximation ratio.
\end{abstract}

\thispagestyle{empty}
\setcounter{page}{0}
\newpage


\section{Introduction}

In this paper, we consider the problem of maximizing a monotone submodular function under a knapsack constraint.
Specifically, given a monotone submodular function $f:2^E \to \bbR_+$ and a weight function $w:E \to [0,1]$, we aim to solve the following optimization problem:
\begin{align*}
  \text{maximize }f(S) \quad \text{subject to } w(S) \leq 1 \text{ and } S \subseteq E,
\end{align*}
where $w(S) = \sum_{e \in S}w(e)$.
This problem has wide applications in machine learning tasks such as sensor placement~\cite{Krause:2008vo}, document summarization~\cite{Lin:2010wpa,Lin:2011wt}, maximum entropy sampling~\cite{Lee:2006cm}, and budget allocation~\cite{Soma:2014tp}.
Although this problem is NP-hard in general, it is known that we can achieve $(1-1/e)$-approximation in polynomial time~\cite{Sviridenko:2004hq}, and this approximation ratio is indeed tight~\cite{Feige:1998gx}.

Although it is useful to know that we can always obtain $(1-1/e)$-approximation in polynomial time, it is observed that a simple greedy method outputs even better solutions in real applications (see, e.g.,~\cite{Krause:2008vo}), and it is more desirable if we can guarantee a better approximation ratio by making assumptions on the input function.
One such assumption is the notion of curvature, introduced by Conforti and Cornu{\'e}jols~\cite{Conforti:1984ig}.
For a monotone submodular function $f:2^E \to \bbR_+$, the \emph{(total) curvature} of $f$ is defined as
\[
  c_f = 1 - \min_{e \in E}\frac{f_{E-e}(e)}{f(e)}.
\]
Intuitively speaking, the curvature measures how close $f$ is to a linear function.
To see this, note that $c_f \in [0,1]$ and $c_f = 0$ if and only if $f$ is a linear function.

It was shown in~\cite{Conforti:1984ig} that, for maximizing a monotone submodular function under a cardinality constraint, the greedy algorithm achieves an approximation ratio $(1-e^{-c_f})/c_f$, and the result was extended to a matroid constraint~\cite{Vondrak:2010zz}.
Recently, Sviridenko~\emph{et~al.}~\cite{Sviridenko:2015ur} obtained a polynomial-time algorithm for a matroid constraint with an approximation ratio $1-c_f/e$, and showed that this approximation ratio is indeed tight for every $c_f \in [0,1]$ even under a cardinality constraint (note that $1-c_f/e$ is strictly larger than $(1-e^{-c_f})/c_f$ except when $c_f = 0$ or $c_f = 1$).

In this paper, we extend these results to a knapsack constraint and present a polynomial-time algorithm under a knapsack constraint with an approximation ratio $1-c_f/e$.
More specifically, we show the following:
\begin{theorem}\label{the:intro}
  There is an algorithm that, given a monotone submodular function $f: 2^E \to \bbR_+$, a weight function $w: E \to [0,1]$, and $\epsilon \in (0,1)$, outputs a (random) set $S \subseteq E$ with $w(S) \leq 1$ (with probability one) satisfying
  \[
    \E[f(S)] \geq \Bigl(1-\frac{c_f}{e} - \epsilon\Bigr)f(O).
  \]
  Here, $O\subseteq E$ is an optimal solution to the problem, i.e., $O$ is a set with $w(O) \leq 1$ that maximizes $f$.
  The running time is $\compintro$, where $n = |E|$.
\end{theorem}
Note that the approximation ratio $1-c_f/e$ is indeed tight for every $c_f \in [0,1]$ because the lower bound given by~\cite{Sviridenko:2015ur} holds even for a cardinality constraint.
To the best of our knowledge, this is the first result for a knapsack constraint that incorporates the curvature to obtain an approximation ratio better than $1-1/e$, which is tight for general submodular functions.


We can apply our algorithm to all the above-mentioned applications to obtain a better solution when the input function has a small curvature.
As a representative example, we consider the budget allocation problem~\cite{Alon:2012em}, which models a marketing process that allocates a given budget among media channels, such as TV, newspapers, and the Web, in order to maximize the impact on customers.
We model the process using a bipartite graph on a vertex set $A \cup B$, where $A$ and $B$ correspond to the media channels and the customers, respectively, and an edge $(a,b) \in A \times B$ represents the potential influence of media channel $a$ on customer $b$.
In a simplified setting where we can use each channel at most once, each media channel $a \in A$ can activate a customer with a predetermined probability $p_a \in [0,1]$.
Then, we have to find a set $S \subseteq A$ that maximizes the expected number of activated customers subject to $w(S) := \sum_{a \in S}w(a) \leq 1$, where $w(a)$ is the cost of using media channel $a$.
We can formulate this problem as the maximization of a monotone submodular function $f:2^A \to \bbR_+$ under the knapsack constraint $w(S) \leq 1$.
We can show that $c_f \leq 1 - \min_{b \in B}p^{|\Gamma(b)|-1}$, where $\Gamma(b)$ denotes the set of neighbors of $b$ in the bipartite graph.
By Theorem~\ref{the:intro}, this immediately gives the approximation ratio $1-c_f/e-\epsilon$ for this problem.
The actual model is more general and discussed in detail in Section~\ref{sec:applications}.

\subsection{Proof technique}

Now, we present the outline of our proof.
Let $f:2^E \to \bbR_+$ be the input function and $O \subseteq E$ be the optimal solution, i.e., $O$ is the set that maximizes $f$ among the sets with weight at most one.
We assume that $c_f = 1-\Omega(\epsilon)$; otherwise, we can use a standard algorithm~\cite{Sviridenko:2004hq} to achieve the desired approximation ratio.
Using the argument in~\cite{Sviridenko:2015ur}, we can decompose the input function $f$ into a monotone submodular function $g:2^E \to \bbR_+$ and a linear function $\ell:2^E \to \bbR_+$ such that, if we can compute a set $S \subseteq E$ with $w(S) \leq 1$ and $f(S) = g(S) + \ell(S) \geq (1-1/e)g(O)+\ell(O)$, then $S$ is a $(1-c_f/e)$-approximate solution.
Moreover, by slightly changing the argument in~\cite{Sviridenko:2015ur}, we can also assume that $c_g=1-\Omega(\epsilon(1-c_f)) = 1-\Omega(\epsilon^2)$.

In order to find the desired set $S \subseteq E$, we use a variant of the continuous greedy algorithm~\cite{Calinescu:2011ju} that simultaneously optimizes $g$ and $\ell$.
In this algorithm, we consider continuous versions of $g$, $\ell$, and $w$, denoted by $G:[0,1]^E \to \bbR_+$, $L:[0,1]^E \to \bbR_+$, and $W:[0,1]^E \to \bbR_+$, respectively.
We note that the function $G$ is called the multilinear extension of $g$, and that $L$ and $W$ are linear functions.
We start with the zero vector $\bx \in [0,1]^E$ and then iteratively update it.
The algorithm consists of $1/\epsilon$ iterations, and roughly speaking, in each iteration, we find a vector $\bv \in [0,1]^E$ with the following properties: (i) $G(\bx + \epsilon \bv) - G(\bx) \geq \epsilon (G(\bx \vee \bone_{O}) - G(\bx))$, (ii) $L(\epsilon \bv) = \epsilon L(\bv) \geq \epsilon \ell(O)$, and (iii) $W(\epsilon \bv) = \epsilon W(\bv)\leq w(O)$. Then, we update $\bx$ by adding $\epsilon \bv$.
Here, $\bone_O$ is the characteristic vector of the set $O$ and $\vee$ is the coordinate-wise maximum.
Intuitively speaking, these conditions mean that moving along the direction $\bv$ from $\bx$ is no worse than moving towards $\bx \vee \bone_O$.
We can find such a vector $\bv$ by linear programming.
Then, after $1/\epsilon$ iterations, we get a vector $\bx \in [0,1]^E$ such that $G(\bx) \geq (1-1/e)G(\bone_O) = (1-1/e)g(O)$, $L(\bx) \geq \ell(O)$, and $W(\bx) \leq w(O)$.
Finally, we obtain a set $S \subseteq E$ by rounding the vector $\bx$, where each element $e \in E$ is added with probability $\bx(e)$.

Unfortunately, this strategy does not work as is.
Here, a crucial issue is that we cannot show the concentration of the weight in the rounding step.
To address this issue, by borrowing an idea from~\cite{Badanidiyuru:2013jc}, we split the elements into large and small ones, where an element is said to be \emph{small} if $g(e) \leq \epsilon^6 g(O)$ and $\ell(e) \leq \epsilon^6 \ell(O)$, and is said to be \emph{large} otherwise (in our analysis, it is more convenient to define large and small elements in terms of $g$ and $\ell$ instead of $w$).
Then, since the curvature of $g$ is bounded away from one, we can bound the number of large elements in $O$ by a function of $\epsilon$.\footnote{Although it is claimed in~\cite{Badanidiyuru:2013jc} that the number of large elements is bounded for any submodular function, it is not true in general.}
Let $O_{\rm L},O_{\rm S} \subseteq O$ be the set of large and small elements in $O$, respectively.
Further, we let $O_{\rm L} = \set{o_1,\ldots,o_m}$.
Then, in each iteration, we do the following:
For each $i \in \set{1,\ldots,m}$, we find an element $e_i$ such that (i) $G(\bx \vee \bone_{e_i}) - G(\bx) \geq G(\bx \vee \bone_{o_i}) - G(\bx)$, (ii) $\ell(e_i) \geq \ell(o_i)$, and (iii) $w(e_i) \leq w(o_i)$. Then, we update $\bx$ by adding $\epsilon\bone_{e_i}$.
Here, $\bone_e$ is a characteristic vector of the element $e\in E$.
Intuitively speaking, adding $e_i$ to the current solution is no worse than  adding $o_i$.
For small items, we find a vector $\bv$ as before by considering the characteristic vector $\bone_{O_{\rm S}}$; then, we update $\bx$ by adding $\epsilon \bv$.
In the rounding step, we handle large and small elements separately.
Note that, for each $i \in \set{1,\ldots,m}$, we have computed $1/\epsilon$ elements (through $1/\epsilon$ iterations).
Then, we chose one of them uniformly at random and add it to the output set.
An advantage of this rounding procedure is that we can guarantee that the chosen element for $i \in \set{1,\ldots,m}$ has weight at most $w(o_i)$.
For small elements, we apply the previous rounding procedure with a minor tweak to guarantee that the output set has weight at most one.

In order to realize this idea, we need to address several additional issues.
First, as we do not know the set $O$, we do not know values related to $O$, such as $g(O)$, $\ell(O)$, $G(\bx \vee \bone_O)$, and $G(\bx \vee \bone_{o_i})$.
Hence, we cannot determine whether an element is small or large, and we cannot find the desired vector or element in each iteration.
We address this issue by guessing these values.
For example, we can show a lower bound and an upper bound on $g(O)$ that are $O(n)$ times apart.
This means that we can find a $(1-\epsilon)$-approximation to $g(O)$ in the geometric sequence of length $O(\log_{1+\epsilon} n)=O((\log n)/\epsilon)$ between the lower and upper bounds.
If we naively guess all the values, as we have $1/\epsilon$ iterations, the resulting time complexity will be $\poly(n) \cdot \bigl((\log n)/\epsilon\bigr)^{\poly(1/\epsilon)}$.
However, since the function $g$ has curvature $1-\Omega(\epsilon^2)$, we can reduce the number of candidate values and thus improve the time complexity to $\poly(n)\cdot (1/\epsilon)^{\poly(1/\epsilon)}$.



\subsection{Related work}

As mentioned earlier, it has been shown that the greedy method achieves $(1-e^{-c_f})/c_f$ approximation for a cardinality constraint~\cite{Conforti:1984ig}.
The result was extended to a matroid constraint by \Vondrak~\cite{Vondrak:2010zz}.
He showed that the result actually holds if we replace $c_f$ with the \emph{curvature to the optimum} $c_f^*$, and the approximation ratio $(1-e^{-c_f^*})/c_f^*$ is tight.
Sviridenko~\emph{et~al.}~\cite{Sviridenko:2015ur} improved the approximation ratio to $1-c_f/e-\epsilon$ for a matroid constraint (and hence, a cardinality constraint), which is unattainable with $c_f^*$, and showed that the approximation ratio $1-c_f/e$ is tight even for a cardinality constraint.

Curvature has been used to explain the empirical performance of the greedy method.
Sharma~\emph{et~al.}~\cite{Sharma:2015ur} considered maximum entropy sampling on Gaussian radial basis functions (RBF) kernels, which can be modeled as the maximization of a monotone submodular function, and showed that the curvature of this problem is close to zero.

The maximization of a submodular function under a knapsack constraint has been studied extensively.
Sviridenko obtained a $(1-1/e)$-approximation algorithm with time complexity $O(n^5)$~\cite{Sviridenko:2004hq}.
We can also obtain $(1-1/e-\epsilon)$-approximation with a constant number of knapsack constraints; however, the time complexity blows up to $n^{\poly(1/\epsilon)}$~\cite{Kulik:2013ix}.
It has been claimed in~\cite{Badanidiyuru:2013jc} that, for any fixed $\epsilon > 0$, there is a $(1-1/e-\epsilon)$-approximation algorithm with time complexity $\tilde{O}(n^2)$.
However, as mentioned in the footnote, their argument has a drawback.
Several approximation guarantees have been achieved in~\cite{Iyer:2013tl} using various parameters of the input function.
However, none of them has an approximation ratio better than $1-1/e$ based solely on the assumption that the curvature is bounded.

\subsection{Organization}
The remainder of this paper is organized as follows.
Section~\ref{sec:pre} introduces the definitions used throughout the paper and reviews the basic properties of submodular functions.
Section~\ref{sec:reduction} explains the reduction to a joint approximation of a monotone submodular function and a monotone linear function.
Section~\ref{sec:algorithm} presents a joint approximation algorithm.
Section~\ref{sec:applications} describes an application to the budget allocation problem.



\section{Preliminaries}\label{sec:pre}
For an integer $n \in \bbN$, let $[n]$ denote the set $\set{1,\ldots,n}$.
In this paper, the symbol $E$ always denotes a (finite) domain of a function.

For a function $w:E \to \bbR$ and a subset $S \subseteq E$, we define $w(S) = \sum_{e \in S}w(e)$.
Similarly, for a vector $\bx \in \bbR^E$ and a set $S \subseteq E$, we define $\bx(S) = \sum_{e \in S}\bx(e)$.
For an element $e \in E$, we define $\bone_e$ as the unit vector whose $e$-th element is $1$.
For a set $S \subseteq E$, we define $\bone_S$ as $\sum_{e \in S}\bone_e$.

Let $f:2^E \to \bbR$ be a function.
For an element $e \in E$, we simply write $f(e)$ to denote $f(\set{e})$.
For a set $S \subseteq E$, we define a function $f_S:2^{E} \to \bbR$ as $f_S(T) = f(S \cup T) - f(S)$.
We say that $f$ is \emph{submodular} if, for any $S,T \subseteq E$,
\[
  f(S) + f(T) \geq f(S \cap T) + f(S \cup T).
\]
An equivalent condition is the \emph{diminishing return property}, which requires $f_S(e) \geq f_T(e)$ for  any $S \subseteq T \subsetneq E$ and $e \in E\setminus T$.
We say that $f$ is \emph{linear} if $f(S) = \sum_{e \in S}f(e)$ holds for every $S \subseteq E$.
Note that, if $f$ is submodular (resp., linear), then $f_S$ is also submodular (resp., linear).

For a vector $\bx \in [0,1]^E$, let $R(\bx)$ denote a random set, where each element $e\in E$ is included in the set with probability $\bx(e)$.
For a submodular function $f:2^E \to \bbR$, the \emph{multilinear extension} $F:[0,1]^E \to \bbR$ of $f$ is defined as
\[
  F(\bx) := \E[f(R(\bx))] = \sum_{S \subseteq E}f(S)\prod_{e \in S} \bx(e) \prod_{e \in E \setminus S}(1-\bx(e)).
\]
For an element $e\in E$ and a vector $\bx \in [0,1]^E$, we define $\partial_e F(\bx)$ as the slope of $F$ at $\bx$ in the direction of $\bone_e$.
The following fact is well known (see, e.g.,~\cite{Feldman:2013Tq}):
\begin{align}
  \partial_e F(\bx) = \frac{F(\bx \vee \bone_e) - F(\bx)}{1-x_e}= \frac{\E[f_{R(\bx)}(e)]}{1-x_e}.\label{eq:partial-derivative}
\end{align}



The following lemma bounds the marginal gain of $F$ when adding $\epsilon \by$ to $\bx$:
\begin{lemma}\label{lem:discretization}
  Let $f:2^E \to \bbR_+$ be a monotone submodular function and $\bx,\by \in [0,1]^E$ be vectors such that $\bx + \epsilon \by \in [0,1]^E$.
  Then,
  \[
    F(\bx + \epsilon \by) - F(\bx) \geq \epsilon \sum_{e\in E}\by(e)\E[f_{R(\bx+\epsilon \by)}(e)].
  \]
\end{lemma}
\begin{proof}
  Let $e_1,\ldots,e_n$ be an arbitrary ordering of elements in $E$.
  For $i \in \set{0,1,\ldots,n}$, let $\bx^i = \bx + \sum_{j=1}^i \epsilon \by(e) \bone_e$.
  Note that $\bx^0 = \bx$ and $\bx^n = \bx+\epsilon \by$.
  Then, we have
  \begin{align*}
    F(\bx + \epsilon \by) - F(\bx)
    & = \sum_{i \in [n]} \partial_{e_i} F(\bx^{i-1})  \epsilon \by(e_i) \tag{by multilinearity of $F$}\\
    & \geq \sum_{i \in [n]} \E[f_{R(\bx^{i-1})}(e)]  \epsilon \by(e_i) \tag{by~\eqref{eq:partial-derivative}}\\
    & \geq \epsilon\sum_{i \in [n]}\by(e_i) \E[f_{R(\bx^n)}(e)].    \tag{by submodularity of $f$}\\
    & = \epsilon\sum_{i \in [n]}\by(e_i) \E[f_{R(\bx+\epsilon \by)}(e)].
    \qedhere
  \end{align*}
\end{proof}


We frequently use the following form of Chernoff's bound.
\begin{lemma}[Relative+Additive Chernoff's bound~\cite{Badanidiyuru:2013jc}]\label{lem:chernoff}
  Let $X_1 , \ldots , X_n$ be independent random variables such that $X_i \in [0,1]$ for every $i \in [n]$.
  Let $X = \frac{1}{n} \sum_{i \in [n]} X_i$ and $\mu = \E[X]$.
  Then, for any $\alpha \in (0,1)$ and $\beta > 0$, we have
  \begin{align*}
    \Pr[|X - \mu| > \alpha \mu+\beta] \leq 2\exp\Bigl(-\frac{n\alpha \beta}{3}\Bigr).
  \end{align*}
\end{lemma}
This immediately gives the following sampling algorithm:
\begin{corollary}\label{cor:chernoff}
  Suppose that we can obtain independent samples of a random variable $X$ bounded in $[0,d]$.
  Let $\mu = \E[X]$.
  Then, there exists an algorithm, denoted by $\textsc{Estimate}_{\alpha,\beta,\delta}(X)$, that, given $\alpha,\beta,\delta\in(0,1)$, outputs a value $\hat{\mu}$ such that
  $|\hat{\mu} - \mu| \leq \alpha \mu+\beta d$
  with probability at least $1-\delta$.
  The number of samples used by the algorithm is $O(\log(1/\delta)/(\alpha\beta))$.
\end{corollary}


\section{Reduction}\label{sec:reduction}

In this section, we prove Theorem~\ref{the:intro} using the following theorem, which gives a joint approximation of a monotone submodular function and a monotone linear function.
\begin{theorem}\label{the:gl}
  There is an algorithm that, given a monotone submodular function $g : 2^E \to \bbR_+$, a monotone linear function $\ell : 2^E \to \bbR_+$, a weight function $w: E \to [0,1]$, and $\epsilon \in (0,1)$, outputs a (random) set $S \subseteq E$ with $w(S) \leq 1$ satisfying
  \[
    \E[g(S) + \ell(S)] \geq \Bigl(1 - \frac{1}{e}\Bigr) g ( O ) + \ell(O) - \epsilon  \bigl(g(O)+\ell(O)\bigr).
  \]
  Here, $O = \arg\max_{T\subseteq E:w(T) \leq 1}\bigl(g(T)+\ell(T)\bigr)$ is an optimal solution.
  The running time is $\compgl$, where $n = |E|$.
\end{theorem}
The proof of Theorem~\ref{the:gl} is given in Section~\ref{sec:algorithm}.
In the remainder of this section, we prove Theorem~\ref{the:intro} using Theorem~\ref{the:gl}.
The argument is similar to that used in~\cite{Sviridenko:2015ur}, but it is more subtle here because the running time in Theorem~\ref{the:gl} depends on the curvature of $g$.
We use the following lemma.
\begin{lemma}[Lemma 2.1 of~\cite{Sviridenko:2015ur}]\label{lem:2.1-of-Sviridenko}
  If $f : 2^E \to \bbR_+$ is a monotone submodular function, then $\sum_{e \in E} f_{E-e}(e) \geq (1 - c_f)f(S)$ for all $S \subseteq E$.
\end{lemma}

\begin{theorem}\label{the:f}
  There is an algorithm that, given a monotone submodular function $f: 2^E \to \bbR_+$, a weight function $w: E \to [0,1]$, and $\epsilon \in (0,1)$, outputs a (random) set $S \subseteq E$ with $w(S) \leq 1$ satisfying
  \[
    \E[f(S)] \geq \Bigl(1-\frac{c_f}{e} - \epsilon\Bigr)f(O).
  \]
  Here, $O = \arg\max_{T \subseteq E: w(T) \leq 1}f(T)$ is an optimal solution.
  The running time is $\compf$, where $n = |E|$.
\end{theorem}
\begin{proof}
  Define the functions $g,\ell:2^E \to \bbR_+$ such that
  \begin{align*}
    \ell(S) = \Bigl(1-\frac{\epsilon}{2}\Bigr)\sum_{e \in S}f_{E-e}(e)
    \quad \text{and} \quad
    g(S)  = f(S) - \ell(S)
  \end{align*}
  for every $S \subseteq E$.

  It is not hard to see that $\ell$ is a nonnegative monotone linear function and that $g$ is a nonnegative monotone submodular function.
  Moreover, the curvature of $g$ is
  \begin{align*}
    c_g
    = 1 - \min_{e \in E}\frac{g_{E-e}(e)}{g(e)}
    = 1 - \min_{e \in E}\frac{f_{E-e}(e) - (1-\epsilon/2)f_{E-e}(e)}{f(e) - (1-\epsilon/2)f_{E-e}(e)}
    \leq 1 - \frac{\epsilon}{2}\min_{e \in E}\frac{f_{E-e}(e)}{f(e)}
    = 1-\frac{\epsilon (1-c_f)}{2}.
  \end{align*}
  Further, Lemma~\ref{lem:2.1-of-Sviridenko} implies that for any set $S \subseteq  E$,
  \[
    \ell(S) = \Bigl(1-\frac{\epsilon}{2}\Bigr)\sum_{e\in S} f_{E-e}(e) \geq \Bigl(1-\frac{\epsilon}{2}\Bigr)(1 - c_f)f(S)
    \geq
    \Bigl(1 - c_f-\frac{\epsilon}{2}\Bigr)f(S).
  \]
  By applying Theorem~\ref{the:gl} to $g$, $\ell$, $w$, and $\epsilon/2$, we can find a (random) set $S \subseteq E$ with $w(S) \leq 1$ satisfying
  \begin{align*}
    \E[f(S)] & = \E[g(S) + \ell(S)]
    \geq  \Bigl(1 - \frac{1}{e}\Bigr)g(O) + \ell(O) - \frac{\epsilon}{2} \bigl( g(O) + \ell(O)\bigr)\\
    & =  \Bigl(1 - \frac{1}{e}\Bigr)  f(O) + \frac{1}{e}\ell(O) - \frac{\epsilon}{2} f(O)\\
    & \geq  \Bigl(1 - \frac{1}{e}\Bigr)  f(O) + \frac{1-c_f-\epsilon/2}{e}f(O) - \frac{\epsilon}{2} f(O)\\
    & \geq  \Bigl(1 - \frac{c_f}{e} - \epsilon\Bigr)  f(O).
  \end{align*}
  The running time is clearly as stated.
\end{proof}

Now, we prove our main theorem.
\begin{proof}[Proof of Theorem~\ref{the:intro}]
  If $c_f < 1-e\epsilon$, then we run the algorithm in Theorem~\ref{the:f}.
  The approximation factor is $1-c_f/e-\epsilon$ and the running time is $\compf=\compfsimple$.

  If $c_f \geq 1-e\epsilon$, then we simply run the $O(n^5)$-time $(1-1/e)$-approximation algorithm presented in~\cite{Sviridenko:2004hq}.
  Then, the approximation factor is
  \[
    1-\frac{1}{e} \geq 1-\frac{1-e\epsilon}{e}-\epsilon  \geq 1-\frac{c_f}{e}-\epsilon.
  \]

  In both cases, the approximation factor is at least $1-c_f/e-\epsilon$ whereas the running time is as desired.
\end{proof}



\section{Proof of Theorem~\ref{the:gl}}\label{sec:algorithm}
In this section, we prove Theorem~\ref{the:gl}.
Throughout this section, $G:[0,1]^E \to \bbR_+$ denotes the multilinear extension of $g$, while $L:[0,1]^E \to \bbR_+$ and $W:[0,1]^E \to \bbR_+$ denote the following linear functions:
\begin{align*}
  L(\bx) = \sum_{e \in E}\bx(e)\ell(e) \quad \text{and} \quad
  W(\bx) = \sum_{e \in E}\bx(e)w(e).
\end{align*}
Furthermore, we define $d_g = \max_{e \in E} g(e)$, $d_\ell = \max_{e \in E} \ell(e)$, and $d_{g,\ell} = \max(d_g,d_\ell)$.
Note that we have $g(O) \leq nd_g \leq nd_{g,\ell}$, $\ell(O) \leq nd_\ell \leq nd_{g,\ell}$, and $d_{g,\ell} \leq g(O)+\ell(O) \leq 2n d_{g,\ell}$.
Recall that $n = |E|$.

In Section~\ref{subsec:small}, we argue that we need to deal with ``small'' and ``large'' elements separately in order to guarantee that we get a set satisfying the knapsack constraint after rounding.
Our algorithm updates a vector $\bx \in [0,1]^E$ in several iterations and then rounds it.
In Section~\ref{subsec:small-continous-greedy}, we present an algorithm that updates $\bx$ by adding a vector supported on small elements.
In Section~\ref{subsec:general-continuous-greedy}, we present the entire algorithm that computes the vector $\bx$ by taking large elements into account.
Then, in Section~\ref{subsec:rounding}, we describe our rounding procedure.
We need to guess several parameters when running the algorithms in Section~\ref{subsec:small-continous-greedy} and~\ref{subsec:general-continuous-greedy}.
Our final algorithm with the guessing process is presented in Section~\ref{subsec:final}.


\subsection{Small elements}\label{subsec:small}

Our algorithm computes a vector $\bx \in [0,1]^E$ with $W(\bx) \leq 1$ and then rounds it to a set.
A natural rounding method is to simply output the random set $R(\bx)$.
Then, we can guarantee that the expected objective values $\E[g(R(\bx))]$ and $\E[ \ell(R(\bx))]$ are sufficiently large and the expected weight $\E[w(R(\bx))]$ is at most one.
However, we cannot guarantee the concentration of $w(R(\bx))$ because some elements have large contributions to the weight.
To resolve this issue, we say that elements in $E$ are \emph{small} if
\[
  g(e) \leq \epsilon^6 g(O) \quad \text{and} \quad \ell(e) \leq \epsilon^6 \ell(O).
\]
Then, we can freely remove some of small elements for decreasing the weight without decreasing the value significantly.
Further, we can prove that the number of large elements is bounded by a polynomial in $\epsilon$ and $c_g$.

An issue here is that we do not know $O$; hence, we cannot determine whether an element is small.
To resolve this issue, we guess the values of $g(O)$ and $\ell(O)$.
Without loss of generality, we can assume that $\epsilon/n = (1-\epsilon)^k$ for some integer $k$; otherwise, we slightly decrease the value of $\epsilon$.
Then, we define a set $V_{\epsilon,n}(g,\ell) = \set{nd_{g,\ell},(1-\epsilon)nd,\ldots,\epsilon d_{g,\ell}, 0}$, and we use the values in $V_{\epsilon}(g,\ell)$ to guess $g(O)$ and $\ell(O)$.
Since $g(O) \leq nd_{g,\ell}$ and $\ell(O) \leq nd_{g,\ell}$ hold, there exist some $v_g, v_\ell \in V_{\epsilon,n}(g,\ell)$ such that
\begin{align}
  (1-\epsilon)v_g - \epsilon d_{g,\ell} \leq g(O) \leq v_g
  \quad \text{and} \quad
  (1-\epsilon)v_\ell - \epsilon d_{g,\ell} \leq \ell(O) \leq v_\ell.
  \label{eq:guess-g(O)-and-l(O)}
\end{align}

We say that an element $e \in E$ is \emph{small with respect to $(v_g,v_\ell)$} if
\[
  g(e) \leq \epsilon^6 v_g
  \quad \text{and} \quad
  \ell(e) \leq \epsilon^6 v_\ell.
\]
Otherwise, we say that an element $e \in E$ is \emph{large with respect to $(v_g,v_\ell)$}.
Let $E_{\rm L}(v_g,v_\ell) \subseteq E$ and $E_{\rm S}(v_g,v_\ell) \subseteq E$ be the sets of large and small elements, respectively, with respect to $(v_g,v_\ell)$.
Further, we define $O_{\rm L}(v_g,v_\ell) = E_{\rm L}(v_g,v_\ell) \cap O$ and $O_{\rm S}(v_g,v_\ell) = E_{\rm S}(v_g,v_\ell) \cap O$.
We omit $v_g$ and $v_\ell$ from these notations if they are clear from the context.

When $v_g$ and $v_\ell$ satisfy~\eqref{eq:guess-g(O)-and-l(O)}, we can upper bound the number of large elements in $O$:
\begin{lemma}\label{lem:bound-on-large-elements}
  If $v_g$ and $v_\ell$ satisfy~\eqref{eq:guess-g(O)-and-l(O)}, then we have $|O_{\rm L}| = O\bigl(\frac{1}{(1-c_g)\epsilon^6}\bigr)$.
\end{lemma}
\begin{proof}
  Let $m_\ell$ be the number of elements $e \in O$ with $\ell(e) > \epsilon^6 v_\ell$.
  Then, we have $\epsilon^6 v_\ell  m_\ell \leq \ell(O)$.
  Since $v_\ell \geq \ell(O)$, we have $m_\ell \leq 1/\epsilon^6$.

  Let $\set{o_1,\ldots,o_{m_g}}$ be the set of elements $e \in O$ with $g(e) > \epsilon^6 v_g$.
  Then, we have
  \[
    g(O)
    \geq
    \sum_{i \in [m_g]}g_{\set{o_1,\ldots,o_{i-1}}}(o_i)
    \geq
    (1-c_g)\sum_{i \in [m_g]}g(o_i)
    \geq
    (1-c_g)\epsilon^6 v_g m_g.
  \]
  Since $v_g \geq g(O)$, we have $m_g \leq \frac{1}{(1-c_g)\epsilon^6}$.

  Then, we have $|O_{\rm L}| \leq m_g + m_\ell = O(\frac{1}{(1-c_g)\epsilon^6})$.
\end{proof}

In addition to the values of $g(O)$ and $\ell(O)$, the value of $|O_{\rm L}|$ is not also not known.
However, we can easily guess it because there are only $O(\frac{1}{(1-c_g)\epsilon^6})$ choices from Lemma~\ref{lem:bound-on-large-elements}.
We use the symbol $m$ to denote the guessed value of $|O_{\rm L}|$.

For each choice of $v_g$, $v_\ell$, and $m$, we compute a (random) set that jointly maximizes $g$ and $\ell$, and the final output is the best one among them.
Since $|V_{\epsilon,n}(g,\ell)| = O(\log_{1/(1-\epsilon)}(n/\epsilon)) = O(\log(n/\epsilon)/\epsilon)$,
this guessing process makes the running time $O\Bigl(\bigl(\log(n/\epsilon)/\epsilon\bigr)^2 \cdot \frac{1}{(1-c_g)\epsilon^6}\Bigr) = O\Bigl(\frac{\log^2(n/\epsilon)}{(1-c_g)\epsilon^8}\Bigr)$ times larger.
The details will be explained in Section~\ref{subsec:final}.



\subsection{Subroutine for handling small elements}\label{subsec:small-continous-greedy}

Here, we explain a subroutine that finds a vector $\bv \in [0,1]^E$ supported on the set $E_S$ of small elements (with respect to the current guesses $v_g$ and $v_\ell$) in order to update the current vector $\bx \in \bbR^E$.
We want $\bv$ to satisfy the following properties: (i) $\sum_{e \in E_{\rm S}}\bv(e) \E[g_{R(\bx)}(e)] \geq \E[g_{R(\bx)}(O_{\rm S})]$,
(ii) $L(\bv) \geq \ell(O_{\rm S})$,
and (iii) $W(\bv) \leq w(O_{\rm S})$.

There are several issues in finding such a vector $\bv$:
We cannot exactly calculate $\E[g_{R(\bx)}(e)]\;(e \in E_{\rm S})$; hence, we need to estimate it.
Further, we do not know the values $\E[g_{R(\bx)}(O_{\rm S})]$ and $\ell(O_{\rm S})$.
In the subroutine presented here, we assume that their guessed values, denoted by  $\gamma$ and $\lambda$, respectively, are given as a part of the input.
Once we succeed in accurately estimating $\E[g_{R(\bx)}(e)]\;(e \in E_{\rm S})$ and the given guessed values $\gamma$ and $\lambda$ are sufficiently accurate, we can find the desired vector $\bv$ by solving a linear program.
A detailed description of the subroutine is given in Algorithm~\ref{alg:small-elements}.

\begin{algorithm}[t!]
  \caption{$\textsc{SmallElements}_{\epsilon,\delta}(g,\ell,w,E_{\rm S}, \gamma,\lambda,\bx)$}\label{alg:small-elements}
  \begin{algorithmic}[1]
  \Require{A monotone submodular function $g:2^E \to \bbR_+$, a monotone linear function $\ell:2^E \to \bbR$, a weight function $w :E \to [0,1]$, $\epsilon,\delta \in (0,1)$, a set of small elements $E_{\rm S}$,
  guessed values $\gamma,\lambda$, and a vector $\bx \in [0,1]^E$.}
  \Ensure{A vector $\bv \in [0,1]^E$.}
  \State{For each $e \in E_{\rm S}$, let $\btheta(e) = \textsc{Estimate}_{\epsilon,\epsilon/n,\delta/n}(g_{R(\bx)}(e))$.}
  \State{Find a vector $\bv \in [0,1]^E$ supported on $E_{\rm S}$ that minimizes $W(\bv)$ subject to
  \[
  \bv \cdot \btheta \geq (1-\epsilon)\gamma - \epsilon d_{g,\ell} \quad \text{and} \quad L(\bv) \geq \lambda.
  \]
  by linear programming.
  }
  \State{\Return $\bv$.}
  \end{algorithmic}
\end{algorithm}

Now, we analyze Algorithm~\ref{alg:small-elements}.
From Corollary~\ref{cor:chernoff} and the fact that $g_{S}(e) \leq d_g \leq d_{g,\ell}$ for every $S \subseteq E$ and $e \in E$, we have the following:
\begin{proposition}\label{pro:small-element-chernoff}
  With probability at least $1-\delta$, we have
  \[
    (1- \epsilon)\E[g_{R(\bx)}(e)] - \frac{\epsilon d_{g,\ell}}{n} \leq \btheta(e) \leq (1+ \epsilon)\E[g_{R(\bx)}(e)] + \frac{\epsilon d_{g,\ell}}{n}
  \] for every $e \in E_{\rm S}$.
\end{proposition}

We formalize the concept that $\gamma$ and $\lambda$ are sufficiently accurate, and then show that Algorithm~\ref{alg:small-elements} outputs a desired vector with accurate $\gamma$ and $\lambda$.
\begin{definition}\label{def:good-guess-small}
  We say that $\gamma$ and $\lambda$ are \emph{good guesses} if
  \[
    \E[g_{R(\bx)}(O_{\rm S})] \geq \gamma \geq (1-\epsilon)\E[g_{R(\bx)}(O_{\rm S})] - \epsilon d_{g,\ell}
    \quad
    \text{and}
    \quad
    \ell(O_{\rm S}) \geq \lambda \geq (1-\epsilon)\ell(O_{\rm S}) - \epsilon d_{g,\ell}
  \] hold, respectively.
\end{definition}
Since $\E[g_{R(\bx)}(O_{\rm S})] \leq nd_{g,\ell}$ and $\ell(O_{\rm S}) \leq nd_{g,\ell}$ hold, we can find good guesses by trying all the values in the set $V_{\epsilon,n}(g,\ell)$.

\begin{lemma}\label{lem:small-element-guess}
  Suppose that $\gamma$ and $\lambda$ are good guesses.
  Then, Algorithm~\ref{alg:small-elements} returns a vector $\bv \in [0,1]^E$ supported on $E_{\rm S}$ such that
  \begin{itemize}
  \item[(i)] $\sum_{e \in E_{\rm S}}\bv(e)\E[g_{R(\bx)}(e)] \geq (1-\epsilon)^3\E[g_{ R(\bx)}(O_{\rm S})] - 3\epsilon d_{g,\ell}$,
  \item[(ii)] $L(\bv) \geq (1-\epsilon)\ell(O_{\rm S})-\epsilon d_{g,\ell}$, and
  \item[(iii)] $W(\bv) \leq w(O_{\rm S})$,
  \end{itemize}
  with probability at least $1-\delta$.
  The time complexity of Algorithm~\ref{alg:small-elements} is $O(\compsmall)$.
\end{lemma}
\begin{proof}
  With probability at least $1-\delta$, the consequence of Proposition~\ref{pro:small-element-chernoff} holds.
  In what follows, we assume that this occurs.

  The vector $\bone_{O_{\rm S}}$ satisfies
  \begin{align*}
    \bone_{O_{\rm S}}\cdot \btheta
    & =
    \sum_{e \in O_{\rm S}}\btheta(e)
    \geq
    \sum_{e \in O_{\rm S}}\Bigl((1-\epsilon)\E[g_{R(\bx)}(e)] - \frac{\epsilon d_{g,\ell}}{n}\Bigr) \\
    & \geq
    (1-\epsilon)\E[g_{R(\bx)}(O_{\rm S})] - \epsilon d_{g,\ell}
    \geq
    (1-\epsilon)\gamma - \epsilon d_{g,\ell}.
  \end{align*}
  Furthermore, we have $L(\bone_{O_{\rm S}}) = \ell(O_{\rm S}) \geq \lambda$.
  Hence, the vector $\bv$ is well defined, and in particular, we have $W(\bv) \leq W(O_{\rm S})$.

  Then, we have
  \begin{align*}
    \sum_{e \in E_{\rm S}}\bv(e)\E[g_{R(\bx)}(e)]
    & \geq
    (1-\epsilon)\sum_{e \in E_{\rm S}}\bv(e) \Bigl(\btheta(e)-\frac{\epsilon d_{g,\ell}}{n}\Bigr)
    \geq
    (1-\epsilon)\sum_{e \in E_{\rm S}}\bv(e) \btheta(e)-\epsilon d_{g,\ell} \\
    & \geq (1-\epsilon)\bigl((1-\epsilon)\gamma -\epsilon d_{g,\ell}\bigr)- \epsilon d_{g,\ell}
    \geq (1-\epsilon)^2\gamma- 2\epsilon d_{g,\ell} \\
    & \geq (1-\epsilon)^2 \bigl((1-\epsilon)\E[g_{R(\bx)}(O_{\rm S})] - \epsilon d_{g,\ell}\bigr) - 2\epsilon d_{g,\ell} \\
    & \geq
    (1-\epsilon)^3\E[g_{R(\bx)}(O_{\rm S})] - 3\epsilon d_{g,\ell}.
  \end{align*}

  It is easy to confirm (ii) and (iii).
  The time complexity for computing $\btheta$ is $O(n^2\log (n/\delta )/\epsilon^2)$ from Corollary~\ref{cor:chernoff}, and the time complexity for solving the linear program is $O(n^4)$ by using the ellipsoid method.
  The total time complexity is bounded by $O(\compsmall)$.
\end{proof}


\subsection{Continuous greedy algorithm with guessing}\label{subsec:general-continuous-greedy}

In this section, we present an algorithm whose goal is to output a vector $\bx \in [0,1]^E$ such that (i) $G(\bx) \geq (1-1/e)g(O)$, (ii) $L(\bx) \geq \ell(O)$, and (iii) $W(\bx) \leq w(O)$.

Our algorithm is a variant of the continuous greedy algorithm~\cite{Calinescu:2011ju} but differs in the following aspects: we consider two functions $g$ and $\ell$, and we  handle large and small elements separately.
Let $m$ be an integer given as a parameter, which is a guessed value of $|O_{\rm L}|$ (with respect to the current values of $v_g$ and $v_\ell$).
Then, we make copies $E_1,\ldots,E_{m}$ of $E$ and define a set $\widehat{E} = \bigcup_{i \in [m]}E_i \cup E_{\rm S}$.
Then, we define a function $\widehat{g} : 2^{\widehat{E}} \rightarrow \bbR_+$ as $\widehat{g}(S_1,S_2,\ldots,S_{m},S_{\rm S}) = g(S_1 \cup \cdots \cup S_{m} \cup S_{\rm S})$.
We note that $\widehat{g}$ is a monotone submodular function.
Let $\widehat{G}$ be the multilinear extension of $\widehat{g}$.

We introduce a vector $\by_i \in [0,1]^E$ for each $i \in [m]$ and another vector $\bz \in [0,1]^E$.
We always guarantee that $\by_i\;(i \in [m])$ is supported on $E_i$ and $\bz$ is supported on $E_{\rm S}$.
Our algorithm runs in $1/\epsilon$ iterations and updates the vectors $\by_i\;(i \in [m])$ and $\bz$ in each iteration.
Here, we assume that $1/\epsilon$ is an integer; otherwise, we slightly decrease $\epsilon$.
The final output is the sequence of vectors $(\by_1,\ldots,\by_{m},\bz)$, and their sum $\bx := \sum_{i \in [m]}\by_i + \bz$ will satisfy the conditions stated initially in this section.
We call the first iteration the iteration at time $0$, the second iteration the iteration at time $\epsilon$, and so on.

To explain how we update the vectors, we introduce several notations.
For $t = \set{0,\epsilon,\ldots,1}$, we define $\by^t_i\;(i \in [m])$ and $\bz^t$ as the vectors $\by_i$ and $\bz$ immediately before the iteration at time $t$.
We note that $\by^0_i = \bzero\;(i \in [m])$ and $\bz^0 = \bzero$ hold.
We define $\by^1_i\;(i \in [m])$ and $\bz^1$ as $\by_i\;(i \in [m])$ and $\bz$, respectively, after the iteration at time $1-\epsilon$.
Note that the algorithm outputs the sequence of vectors $(\by^1_1,\ldots,\by^1_m,\bz^1)$.
Then, we define $\bx^t = \sum_{i \in [m]}\by^t_i + \bz^t$.
Further, for $t \in \set{0,\epsilon,\ldots,1-\epsilon}$ and $i \in \set{0,1,\ldots,m}$, we define $\bx^t_i = \sum_{j \leq i}\by^{t+ \epsilon}_j + \sum_{j > i}\by^t_j + \bz^t$, i.e., the vector obtained after the iteration at time $t-\epsilon$ followed by updating $\by_1,\ldots,\by_i$.
Note that $\bx^t_0 = \bx^t$.

As in the argument in Section~\ref{subsec:small}, we try all possible values for guessing $|O_{\rm L}|$.
Hence, in what follows, we assume that the guessed value $m$ is correct, i.e., $m = |O_{\rm L}|$.

Let $o_1,\ldots,o_{m}$ be the large elements in $O$, i.e., $O_{\rm L} = \set{o_1,\ldots,o_m}$.
For $i \in [m]$, let $\widehat{o}_i$ be the copy of $o_i$ in $E_i$.
Then, we define $\widehat{O}_{\rm L} = \set{\widehat{o}_1,\ldots,\widehat{o}_m}$ and $\widehat{O} = \widehat{O}_{\rm L} \cup O_{\rm S} \subseteq \widehat{E}$.
For each $i \in [m]$, we update the vector $\by_i^t$ to $\by_i^{t+\epsilon}$ by finding an element $e_i^t \in E_i$ and adding the vector $\epsilon \bone_{e_i^t}$.
Here, we want the element $e_i^t$ to satisfy (i) $\E[\widehat{g}_{R(\bx^t_{i-1})}(e_i^t)] \geq \E[\widehat{g}_{R(\bx^t_{i-1})}(\widehat{o}_i)]$, (ii) $\ell(e_i^t) \geq \ell(o_i)$, and (iii) $w(e_i^t) \leq w(o_i)$.
As we do not know the values of $\E[\widehat{g}_{R(\bx^t_{i-1})}(\widehat{o}_i)]$ and $\ell(o_i)$, the algorithm requires their guessed values $\gamma_i^t$ and $\lambda_i$, respectively.
We do not have to guess $w(o_i)$ because we will choose the element with the minimum weight satisfying (i) and (ii).

Then, we update the vector $\bz^t$ to $\bz^{t+\epsilon}$ by finding a vector $\bv^t$ and adding the vector $\epsilon \bv^t$.
Here, we want the vector $\bv$ to satisfy (i) $\sum_{e \in E_{\rm S}}\bv(e)\E[\widehat{g}_{R(\bx)}(e)]  \geq  \E[\widehat{g}_{R(\bx^t_{i-1})}(O_S)]$ (note that $O_S \subseteq E_S \subseteq \widehat{E}$), (ii) $L(\bv) \geq \ell(O_{\rm S})$, and (iii) $W(\bv) \leq w(O_{\rm S})$.
Such a vector can be found by calling \textsc{SmallElements} with the guessed values $\gamma^t_{\rm S}$ and $\lambda_{\rm S}$ for $\E[\widehat{g}_{R(\bx^t_{i-1})}(O_S)]$ and $\ell(O_{\rm S})$, respectively.
A detailed description of the algorithm is given in Algorithm~\ref{alg:greedy-general}.

We will show that $\widehat{G}(\bx^{t+\epsilon})-G(\bx^t) \geq \epsilon (g(O) - \widehat{G}(\bx^{t+\epsilon}))$, which is sufficient to show that $\widehat{G}(\bx^{1})$ is close to the $(1-1/e)$-approximation to $g(O)$.
We will also show that $L(\bx^1) \geq \ell(O)$ and $W(\bx) \leq w(O)$.

\begin{algorithm}[t!]
  \caption{$\textsc{GuessingContinuousGreedy}_{\epsilon,\delta}(g,\ell,w,E_{\rm L}, E_{\rm S},m,\set{\gamma^t_i}, \set{\gamma^t_{\rm S}}, \set{\lambda_i},\lambda_{\rm S})$}\label{alg:greedy-general}
  \begin{algorithmic}[1]
    \Require{A monotone submodular function $g:2^E \to \bbR_+$, a monotone linear function $\ell:2^E \to \bbR$, a weight function $w :E \to [0,1]$, $\epsilon,\delta\in (0,1)$, a set of large and small elements $E_{\rm L}$ and $E_{\rm S}$, an integer $m\in \bbN$, guessed values $\set{\gamma^t_i}_{i\in [m],t \in \set{0,\epsilon,\ldots,1-\epsilon}}$, $\set{\gamma^t_{\rm S}}_{t \in \set{0,\epsilon,\ldots,1-\epsilon}}$, $\set{\lambda_i}_{i \in [m]}$, and $\lambda_{\rm S}$.}
    \Ensure{A vector $\bx \in [0,1]^E$.}
    \State{$\by_i\leftarrow \bzero\in [0,1]^{\widehat{E}}$ for $i \in [m]$ and $\bz\leftarrow \bzero\in [0,1]^{\widehat{E}}$.}
    \For{$(t\leftarrow 0$; $t \leq 1-\epsilon$; $t\leftarrow t+\epsilon$)}
       \For{($i \leftarrow  1$; $i \leq  m$; $i \leftarrow  i + 1$)}
        \State{$\btheta_i(e) \leftarrow \textsc{Estimate}_{\epsilon,\epsilon/m,\epsilon\delta/(2nm)}(\E[\widehat{g}_{R(\bx)}](e))$ for each $e \in E_i$.}
        \State{Let $e = \arg\min\set{w(e) \mid e \in E_i, \btheta_i(e) \geq (1-\epsilon)\gamma^t_i - \epsilon d/m, \ell(e) \geq \lambda_i}$.}
        \State{$\by_i \leftarrow \by_i + \epsilon \bone_e$.}
      \EndFor
      \State{$\bv \leftarrow \textsc{SmallElements}_{\epsilon,\epsilon\delta/2}(\widehat{g},\ell,w,E_{\rm S},\gamma^t_{\rm S},\lambda_{\rm S},\sum_{i\in[m]}\by_i + \bz)$.}
      \State{$\bz \leftarrow \bz + \epsilon \bv$.}
    \EndFor
    \State{\Return $(\by_1,\ldots,\by_{m},\bz)$.}
  \end{algorithmic}
\end{algorithm}

For $t \in \set{0,\epsilon,\ldots,1-\epsilon}$ and $i \in [m]$, let $\btheta_i^t$ be the $\btheta_i$ used in the iteration at time $t$.
From Lemma~\ref{lem:chernoff} and the union bound, we immediately have the following:
\begin{proposition}\label{pro:general-chernoff}
  With probability at least $1-\delta/2$, we have
  \begin{align*}
    (1- \epsilon )\E[\widehat{g}_{R(\bx^{t}_{i-1})}(e)] - \frac{\epsilon d_{g,\ell}}{m} & \leq \btheta^t_i(e)  \leq (1 + \epsilon )\E[\widehat{g}_{R(\bx^{t}_{i-1})}(e)] + \frac{\epsilon d_{g,\ell}}{m}
  \end{align*}
  for every $t \in \set{0,\epsilon,\ldots,1-\epsilon}$, $i \in [m]$, and $e \in E_i$.
\end{proposition}

We formalize the concept that $\gamma_i^t$ and $\lambda_i$ are sufficiently accurate.
\begin{definition}\label{def:good-guess-large}
  For $t \in \set{0,\epsilon,\ldots,1-\epsilon}$ and $i \in [m]$, we say that $\gamma^t_i$ is a \emph{good guess} if
  \[
    \E[\widehat{g}_{R(\bx^{t}_{i-1})}(\widehat{o}_i)] \geq \gamma^t_i \geq (1-\epsilon) \E[\widehat{g}_{R(\bx^{t}_{i-1})}(\widehat{o}_i)] - \frac{\epsilon d_{g,\ell}}{m}
  \] holds.
  For $i \in [m]$, we say that $\lambda_i$ is a \emph{good guess} if
  \[
    \ell(o_i) \geq \lambda_i \geq (1-\epsilon) \ell(o_i) - \frac{\epsilon d_{g,\ell}}{m}
  \] holds.
\end{definition}
Since $\E[\widehat{g}_{R(\bx^{t}_{i-1})}(o_i)] \leq d_{g,\ell}$ and $\ell(o_i) \leq d_{g,\ell}$ hold for every $i \in [m]$, we can find good guesses by trying all the values in the set $V_{\epsilon,m}(g,\ell)/m := \set{v / m\mid v \in V_{\epsilon,m}(g,\ell)}$.

\begin{lemma}\label{lem:large-guess}
  Suppose that the consequence of Proposition~\ref{pro:general-chernoff} holds and that $\set{\gamma_i^t}$ and $\set{\lambda_i}$ are good guesses.
  Then,
  for every $t \in \set{0,\epsilon,\ldots,1-\epsilon}$, we have the following:
  \begin{itemize}
  \item[(i)] $\E[\widehat{g}_{R(\bx^{t}_{i-1})}(e^t_i)] \geq  (1 - \epsilon)^3\E[\widehat{g}_{R(\bx^{t}_{i-1})}(\widehat{o}_i)] - 3\epsilon d_{g,\ell}/m$,
  \item[(ii)] $\ell(e^t_i) \geq  (1-\epsilon)\ell(o_i) - \epsilon d_{g,\ell}/m$ for $i \in [m]$, and
  \item[(iii)] $w(e^t_i) \leq w(o_i)$ for $i \in [m]$.
  \end{itemize}
\end{lemma}
\begin{proof}
  Fix $t \in \set{0,\epsilon,\ldots,1-\epsilon}$ and $i \in [m]$.
  Note that we have
  \[
    \btheta^t_i(o_i)
    \geq
    (1-\epsilon)\E[\widehat{g}_{R(\bx^{t}_{i-1})}(\widehat{o}_i)]-\frac{\epsilon d_{g,\ell}}{m}
    \geq
    (1-\epsilon)\gamma^t_i- \frac{\epsilon d_{g,\ell}}{m}
  \] and
  $\ell(o_i) \geq \lambda_i$.
  Since $o_i$ (in $E_i$) is a candidate for $e^t_i$, the element $e^t_i$ is well defined.
  In particular, we have $w(e^t_i) \leq w(o_i)$ because $e^t_i$ is chosen as the element with the minimum element satisfying the conditions.

  We have
  \[
    (1+\epsilon)\E[\widehat{g}_{R(\bx^{t}_{i-1})}(e^t_i)]+\frac{\epsilon d_{g,\ell}}{m}
    \geq
    \btheta^t_i(e^t_i) \geq
    (1-\epsilon)\gamma^t_i - \frac{\epsilon d_{g,\ell}}{m}
    \geq
    (1-\epsilon)^2\E[\widehat{g}_{R(\bx^t_i)}(o_i)]-\frac{2\epsilon d_{g,\ell}}{m}.
  \]
  Rearranging this inequality, we get (i).
  Further, (ii) is immediate from the fact that $\ell(e_i^t) \geq \lambda_i$.
\end{proof}

We say that $\gamma_{\rm S}^t$ for $t \in \set{0,\epsilon,\ldots,1-\epsilon}$ and $\lambda_{\rm S}$ are \emph{good guesses} if they are good guesses in the sense of Definition~\ref{def:good-guess-small}.
Then, we have the following:
\begin{lemma}\label{lem:continuous-greedy}
  Suppose that $\set{\gamma_i^t},\set{\lambda_i}$, $\set{\gamma_{\rm S}^t}$, and $\lambda_{\rm S}$ are good guesses.
  Then, Algorithm~\ref{alg:greedy-general} returns vectors $\by_1,\ldots,\by_{m},\bz$ such that $\bx := \sum_{i \in [m]}\by_i + \bz$ satisfies the following:
  \begin{itemize}
  \item[(i)] $\widehat{G}(\bx) \geq (1-1/e-O(\epsilon))g(O)-6\epsilon d_{g,\ell}$,
  \item[(ii)] $L(\bx) \geq (1-O(\epsilon))\ell(O)- 2 \epsilon d_{g,\ell}$, and
  \item[(iii)] $W(\bx) \leq w(O)$,
  \end{itemize}
  with probability at least $1-\delta$.
  The running time is $O(\compgeneral)$.
\end{lemma}
\begin{proof}
  With probability $1-\delta/2$, the consequence of Proposition~\ref{pro:general-chernoff} holds.
  Further, with probability $1-\delta/2$, all the invocations of \textsc{SmallElements} succeed in outputting vectors with the guarantees in Lemma~\ref{lem:small-element-guess}.
  By the union bound, all these occur with probability at least $1-\delta$.
  In what follows, we assume that this occurs.

  First, we check~(i).
  For each $t \in \set{0,\epsilon,\ldots,1-\epsilon}$, we have
  \begin{align*}
    & \widehat{G}(\bx^t_{m}) - \widehat{G}(\bx^t)
    = \sum_{j=1}^{m}\Bigl(\widehat{G}(\bx^t_{j-1} + \epsilon \bone_{e^t_j}) - \widehat{G}(\bx^t_{j-1})\Bigr)\\
    & \geq \epsilon\sum_{j=1}^{m}\Bigl(\E[\widehat{g}_{R(\bx^t_{j-1})}(e^t_j)]\Bigr) \tag{By concavity of $\widehat{G}$}\\
    & \geq \epsilon(1 - \epsilon)^3\sum_{j=1}^{m}\Bigl(\E[\widehat{g}_{R(\bx^t_{j-1})}(\widehat{o}_j)] - \frac{3\epsilon d_{g,\ell}}{m}\Bigr) \tag{By (i) of Lemma~\ref{lem:large-guess}}\\
    & \geq \epsilon(1 - \epsilon)^3\sum_{j=1}^{m} \Bigl(\E\bigl[\widehat{g}_{R(\bx^t_{m}) \cup \set{\widehat{o}_k \mid k \in [j-1] }}(\widehat{o}_j) \bigr]\Bigr) -3\epsilon^2  d_{g,\ell} \\
    & = \epsilon(1 - \epsilon)^3\Bigl(\E\bigl[\widehat{g}(R(\bx^{t}_{m}) \cup \widehat{O}_{\rm L}) - \widehat{g}(R(\bx^{t}_{m}))\bigr]\Bigr)- 3\epsilon^2 d_{g,\ell} \\
    & \geq \epsilon(1 - \epsilon)^3(\E[\widehat{g}_{R(\bx^t_{m})}(\widehat{O}_{\rm L})] ) - 3\epsilon^2 d_{g,\ell}.
  \end{align*}

  For each $t \in \set{0,\epsilon,\ldots,1-\epsilon}$, we have
  \begin{align*}
    & \widehat{G}(\bx^{t+\epsilon}) - \widehat{G}(\bx^t_{m})
    =
    \widehat{G}(\bx^{t}_{m} + \epsilon \bv) - \widehat{G}(\bx^t_{m}) \\
    & \geq
    \epsilon \sum_{e \in E}\bv(e) \E[\widehat{g}_{R(\bx^{t+\epsilon})}(e)] \tag{by Lemma~\ref{lem:discretization}} \\
    & \geq
    \epsilon \bigl((1-\epsilon)^3 \E[\widehat{g}_{R(\bx^{t+\epsilon})}(O_{\rm S})] - 3 \epsilon d_{g,\ell} \bigr) \tag{by (i) of Lemma~\ref{lem:small-element-guess}}\\
    & \geq
    \epsilon (1-\epsilon)^3 \E[\widehat{g}_{R(\bx^{t+\epsilon})}(O_{\rm S})] - 3\epsilon^2 d_{g,\ell}.
  \end{align*}

  Combining these two inequalities, we get
  \begin{align*}
    & \widehat{G}(\bx^{t+\epsilon}) - \widehat{G}(\bx^t) \\
    & \geq
    \epsilon(1 - \epsilon)^3(\E[\widehat{g}_{R(\bx^t_{m})}(\widehat{O}_{\rm L})] )- 3\epsilon^2 d_{g,\ell}
    +
    \epsilon (1-\epsilon)^3 \E[\widehat{g}_{R(\bx^{t+\epsilon})}(O_{\rm S} )]  - 3 \epsilon^2 d_{g,\ell}\\
    & \geq
    \epsilon (1 - \epsilon)^3\E[\widehat{g}_{R(\bx^{t+\epsilon})}(\widehat{O})] - 6 \epsilon^2 d_{g,\ell}\\
    & \geq
    \epsilon (1 - \epsilon)^3\bigl(\widehat{g}(\widehat{O}) - \widehat{G}(\bx^{t+\epsilon})\bigr) - 6\epsilon^2 d_{g,\ell}
    =
    \epsilon (1 - \epsilon)^3\bigl(g(O) - \widehat{G}(\bx^{t+\epsilon})\bigr) - 6\epsilon^2 d_{g,\ell}.
  \end{align*}

  Rewriting the above inequality, we get
  \[
    g(O) - \frac{\beta}{\alpha} - \widehat{G}(\bx^{t + \epsilon})  \leq \frac{g(O) - \frac{\beta}{\alpha} - \widehat{G}(\bx^t)}{1+\alpha},
  \]
  where $\alpha = \epsilon(1 - \epsilon)^3$ and $\beta = 6\epsilon^2 d_{g,\ell}$.
  Then, by induction, we can prove that
  \[
    g(O) - \frac{\beta}{\alpha} - \widehat{G}(\bx^t) \leq \frac{1}{(1+\alpha)^{t/\epsilon}}\Bigl(g(O) - \frac{\beta}{\alpha}\Bigr).
  \]
  Substituting $t = 1$ and rewriting once again, we get
  \[
    \widehat{G}(\bx^1) \geq \Bigl(1 - \frac{1}{(1+\alpha)^{1/\epsilon}}\Bigr)(g(O)- \frac{\beta}{\alpha})
    \geq \Bigl(1-\frac{1}{e}-O(\epsilon )\Bigr)\Bigl(g(O)-\frac{\beta}{\alpha}\Bigr)
    \geq \Bigl(1-\frac{1}{e}-O(\epsilon )\Bigr)g(O) - 6\epsilon d_{g,\ell},
  \]
  assuming that $\epsilon$ is sufficiently small, say, less than $1/2$.

  To see (ii), we have for $t \in \set{0,\ldots,1-\epsilon}$ that
  \begin{align*}
    L(\bx^{t+\epsilon}) - L(\bx^{t})
    & = L\Bigl(\bx^t + \epsilon \sum_{i=1}^{m} \bone_{e^t_i} + \epsilon \bv\Bigr) - L(\bx^t) \\
    & = \epsilon \sum_{i=1}^{m}\ell(e^t_i) + \epsilon (1-\epsilon)\ell(O_{\rm S}) - \epsilon^2 d_{g,\ell} \tag{by (ii) of Lemma~\ref{lem:small-element-guess}}\\
    & \geq \epsilon((1-\epsilon)\ell(O_{\rm L})-\epsilon d_{g,\ell}) + \epsilon (1-\epsilon)\ell(O_{\rm S}) - \epsilon^2 d_{g,\ell}  \tag{by (ii) of Lemma~\ref{lem:large-guess}}\\
    & \geq \epsilon(1-\epsilon)\ell(O)- 2\epsilon^2 d_{g,\ell}.
  \end{align*}
  By induction, we get $L(\bx^1) \geq (1-\epsilon)\ell(O) - 2\epsilon d_{g,\ell}$.

  To see (iii), we have for $t \in \set{0,\ldots,1-\epsilon}$ that
  \begin{align*}
    W(\bx^{t+\epsilon}) - W(\bx^t)
    & = W\Bigl(\bx^t + \epsilon \sum_{i=1}^{m} \bone_{e^t_i} + \epsilon \bv\Bigr) - W(\bx^t) \\
    & \leq \epsilon (w(O_{\rm L})+ w(O_{\rm S}) ) \tag{By (iii)'s of Lemmas~\ref{lem:small-element-guess} and~\ref{lem:large-guess}}\\
    & = \epsilon w(O).
  \end{align*}
  By induction, we get $W(\bx^1) \leq w(O)$.

  Finally, we analyze the time complexity.
  For estimating the $\btheta_i$'s, we need $O(\frac{nm}{\epsilon} \cdot \frac{m}{\epsilon^2}\log \frac{nm}{\epsilon \delta}) = O(\frac{nm^2}{\epsilon^3}\log \frac{n m}{\epsilon \delta})$ time.
  The time complexity of \textsc{SmallElements} is at most $O(\frac{1}{\epsilon} \cdot \bigl(n^4+\frac{n^2}{\epsilon^2}\log \frac{1}{\epsilon \delta}) = O(\frac{n^4}{\epsilon}+\frac{n^2}{\epsilon^3}\log \frac{1}{\epsilon \delta})$.
  Hence, the running time is as desired.
\end{proof}



\subsection{Rounding}\label{subsec:rounding}
In this section, we explain how to round the vectors obtained by \textsc{GuessingContinuousGreedy} (Algorithm~\ref{alg:greedy-general}).

Let $(\by_1,\ldots,\by_m,\bz)$ be the vectors obtained by \textsc{GuessingContinuousGreedy}, and let  $\bv^{t}$ be the vector supported on $E_{\rm S}$ obtained in the iteration at time $t$ in \textsc{GuessingContinuousGreedy}.
Note that $\bz = \sum_{t \in \set{0,\epsilon,\ldots,1-\epsilon}}\bv^t$.
Our algorithm is summarized in Algorithm~\ref{alg:rounding}.

\begin{algorithm}[t!]
  \caption{$\textsc{Rounding}_{\epsilon}(w,E_{\rm L}, E_{\rm S},m,\set{\by_i},\bz)$}\label{alg:rounding}
  \begin{algorithmic}[1]
    \Require{A weight function $w:E \to [0,1]$, $E_{\rm _L},E_{\rm S} \subseteq E$, a set of vectors $\set{\by_i}_{i \in [m]}$, and a vector $\bz$.}
    \Ensure{A set $S \subseteq E$.}
    \State{$S_{\rm L} \leftarrow \emptyset$, $S_{\rm S} \leftarrow \emptyset$.}
    \State{Define $\bz' \in [0,1]^{E_{\rm S}}$ as $\bz'(e) = (1-\epsilon)\bz(e)$ if $w(e) < \epsilon^3 \max_t W(\bv^t )$ and $\bz'(e) = 0$ otherwise.}
    \State{For each $e \in E_{\rm S}$, add it to $S_{\rm S}$ independently with probability $\bz'(e)$.}
    \For{$i \in [m]$}
      \State{Add exactly one element in $E_i$ to $S_{\rm L}$ (as an element of $E$), where an element $e \in E_i$ is chosen with probability $\by_i(e)$.}
    \EndFor
    \If{$w(S_{\rm S} \cup S_{\rm L}) \leq 1$}
      \State{\Return $S_{\rm S} \cup S_{\rm L}$.}
    \Else
      \State{\Return $\emptyset$.}
    \EndIf
  \end{algorithmic}
\end{algorithm}

We use the following lemma to analyze the objective value of the output set.
\begin{lemma}[Lemma~3.7 of~\cite{Calinescu:2011ju}]\label{lem:single-rounding}
  Let $E = E_1 \cup \cdots \cup E_k$, let $f : 2^E \to \bbR_+$ be a monotone submodular function, and for all $i \neq j$, we have $E_i \cap E_j = \emptyset$.
  Let $\bx \in \bbR_+^E$ such that for each $E_i$ we have $\bx(E_i) \leq 1$.
  If $T$ is a random set where we sample independently from each $E_i$ at most one random element, i.e., element $e$ with probability $\bx(e)$, then
  \[
    \E[f(T)] \geq F(\bx).
  \]
\end{lemma}

\begin{lemma}\label{lem:rounding-objective-values}
  We have $\E[g(S_{\rm L} \cup S_{\rm S})] \geq (1 - \epsilon)\widehat{G}(\bx) - \epsilon^3 v_g$ and $\E[\ell(S_{\rm L} \cup S_{\rm S})] \geq (1 - \epsilon)L(\bx) - \epsilon^3 v_\ell$.
\end{lemma}
\begin{proof}
  Let $\bx' = \sum_{i\in [m]}\by_i + \bz'$.
  First, let us relate the value of the vector $\bx$ to that of $\bx'$.
  \begin{align*}
    \widehat{G}(\bx')=\widehat{G}\Bigl(\sum_{i \in [m]} \by_i + \bz' \Bigr)
    & \geq \widehat{G}\Bigl(\sum_{i \in [m]} \by_i + (1-\epsilon)\bz\Bigr) - \sum_{e \in E_{\rm S}: w(e) \geq \epsilon^3 \max_t W(\bv^t)}\bz(e)g(e) \\
    & \geq \widehat{G}\Bigl((1-\epsilon)(\sum_{i \in [m]} \by_i+\bz)\Bigr) - \max_{e\in E_{\rm S}} g(e) \sum_{e \in E_{\rm S}:w(e) \geq \epsilon^3 \max_t W(\bv^t)}\bz(e) \\
    & \geq (1-\epsilon)\widehat{G}(\bx)- \epsilon^6 v_g \cdot \frac{1}{\epsilon^3}
    \geq (1-\epsilon)\widehat{G}(\bx)-\epsilon^3 v_g.
  \end{align*}

  Next, we note that we get $S_{\rm L}$ by selecting exactly one random element from each $E_i$, which is a copy of $E_{\rm L}$, and we get $S_{\rm S}$ by sampling independently from $\bv$.
  Hence, by applying Lemma~\ref{lem:single-rounding} with sets $E_1,\ldots,E_{m}$ and sets $\set{\set{e} \mid e \in S_{\rm S}}$, we get
  \[
    \E[g(S_{\rm L} \cup S_{\rm S})]
    \geq (1-\epsilon)\widehat{G}(\bx) - \epsilon^3 v_g.
  \]

  By a similar argument, we get $L(\bx') \geq (1-\epsilon)L(\bx) - \epsilon^3 v_\ell$, and we have $\E[\ell(S_{\rm L} \cup S_{\rm S})] \geq (1-\epsilon)L(\bx)-\epsilon^3 v_\ell$.
\end{proof}

Next, we show that the probability that the weight of the output set exceeds $w(O)$ decays exponentially.
\begin{lemma}\label{lem:rounding-weight}
  For any $\gamma \geq 1$, we have $w(S_{\rm L} \cup S_{\rm S}) \leq \gamma w(O)$ with probability $1-\exp\bigl(-\Omega(\gamma/\epsilon^2)\bigr)$.
\end{lemma}
\begin{proof}
  Recall that, for each $i \in [m]$, the vector $\by_i$ is the sum of $1/\epsilon$ elements $e^0_i,\ldots,e^{1-\epsilon}_i$, and we pick one of them in Algorithm~\ref{alg:rounding}.
  By the condition $w(e^t_i) \leq w(o_i)$ for every $i \in [m]$ and $t \in \set{0,\epsilon,\ldots,1-\epsilon}$, the weight of the large elements after the rounding will be less than that of the large elements of the optimal solution.
  Hence, it is sufficient to prove that $w(S_{\rm S} ) \leq  \gamma w(O_{\rm S})$ holds with probability $1-\exp(-\Omega(\gamma/\epsilon^2))$, where $S_{\rm S}$ is the set obtained by rounding $\bz'$.

  First, note that
  \[
    \E[w(S_{\rm S})] = \E[w(R(\bz'))] \leq (1-\epsilon)\E[w(R(\bz))] \leq (1-\epsilon)\max_{t} W(\bv^t) \leq (1-\epsilon)w(O_{\rm S}).
  \]

  For each $e \in E$, we set up a random variable $X_e$ to be $X_e = w(e)/(\epsilon^3w(O_{\rm S}))$ if $e \in S_{\rm S}$ and $X_e = 0$ otherwise.
  Note that each $X_e$ is bounded in $[0,1]$ because $\max_t W(\bv^t) \leq w(O_{\rm S})$.
  For $X = \sum_{e \in E_{\rm S}}X_e$, we have $\mu := \E[X] = \E[w(S_{\rm S})]/(\epsilon^3w(O_{\rm S})) \leq (1-\epsilon)/\epsilon^3$.

  Invoking Lemma~\ref{lem:chernoff} with $\alpha=\epsilon/2$ and $\beta=\gamma/(2\epsilon^3)$, we have
  \begin{align*}
    \Pr[w(S_{\rm S}) > \gamma w(O_{\rm S})]
    & = \Pr\Bigl[X \geq \frac{\gamma}{\epsilon^3}\Bigr]
    \leq \Pr\Bigl[X \geq (1+\alpha)\mu + \beta\Bigr] \\
    & \leq
    2\exp\Bigl(-\frac{\alpha \beta}{3}\Bigr)
    =
    \exp\Bigl(-\Omega\Bigl(\frac{\gamma}{\epsilon^2}\Bigr)\Bigr).
    \qedhere
  \end{align*}
\end{proof}

\begin{lemma}\label{lem:rounding}
  Algorithm~\ref{alg:rounding} outputs a (random) set $S$ with $w(S) \leq 1$ satisfying
  \[
  \E[g(S)+\ell(S)] \geq (1 - \epsilon)(\widehat{G}(\bx)+L(\bx)) - O(\epsilon) \cdot (g(O)+\ell(O)+v_g+v_\ell).
  \]
\end{lemma}
\begin{proof}
  It is clear that we always have $w(S) \leq 1$.

  Now, we analyze the objective value attained by $S$.
  For any $\gamma \geq 1$, the probability that $w(S_{\rm L} \cup S_{\rm S}) > \gamma w(O)$ is at most $\exp\bigl(-C\gamma/\epsilon^2\bigr)$ for some $C > 0$ by Lemma~\ref{lem:rounding-weight}.
  Note that, if $T \subseteq E$ satisfies $w(T) \leq \gamma w(O)$, then $g(T)+\ell(T) \leq \gamma (g(O)+\ell(O))$ from the submodularity of $g+\ell$.
  By Lemma~\ref{lem:rounding-objective-values}, we have
  \begin{align*}
    & \E[g(S) + \ell(S)]\\
    & \geq
    \E[g(S_{\rm L} \cup S_{\rm S})+\ell(S_{\rm L} \cup S_{\rm S})]
    -
    \int_1^\infty \gamma(g(O)+\ell(O))\exp(-C\gamma/\epsilon^2) \mathrm{d}\gamma \\
    & \geq
    (1 - \epsilon)\widehat{G}(\bx) - \epsilon^3 v_g
    +
    (1 - \epsilon)L(\bx) - \epsilon^3 v_\ell
    - \frac{\epsilon^4+C\epsilon^2}{C^2}\exp(-C/\epsilon^2)(g(O)+\ell(O))\\
    & =
    (1 - \epsilon)(\widehat{G}(\bx)+L(\bx))
    - O(\epsilon) \cdot (g(O)+\ell(O)+v_g+v_\ell).
    \qedhere
  \end{align*}
\end{proof}


\subsection{Putting things together}\label{subsec:final}

Now, we present our entire algorithm.
The idea is to simply guess $v_g$, $v_\ell$, $m$, $\set{\gamma^t_i}$, $\set{\lambda_i}$, $\set{\gamma^t_{\rm S}}$, and $\lambda_{\rm S}$, run Algorithm~\ref{alg:greedy-general} with the guessed values, and then round the obtained vectors using Algorithm~\ref{alg:rounding}.



Naively, we have $O(|V_{\epsilon,n}(g,\ell)|^{O(1/\epsilon)})=O((\log (n/\epsilon)/\epsilon)^{O(1/\epsilon)} )$ choices for the sequence $\set{\gamma^t_{\rm S}}$.
We can decrease the number of choices since $g$ has a bounded curvature.
If we have a guess $\gamma_{\rm S}^0$ such that $\gamma_{\rm S}^0 \geq g(O_{\rm S}) \geq (1-\epsilon)\gamma_{\rm S}^0$, then we must have $\gamma_{\rm S}^0 \geq g_{S}(O_{\rm S}) \geq (1-\epsilon)(1-c_g) \gamma_{\rm S}^0$ for any set $S \subseteq E$.
Hence, it suffices to consider sequences whose maximum and minimum values are within a factor of $(1-\epsilon)(1-c_g)$.
Let $V_{\epsilon,n,\gamma_{\rm S}^0}(g,\ell) := \set{v \in V_{\epsilon,n}(g,\ell) \mid v \geq (1-\epsilon)(1-c_g)\gamma_{\rm_S}^0}$.
Then, the number of such sequences is at most $|V_{\epsilon,n}(g,\ell)| \cdot |V_{\epsilon,n,\gamma_{\rm S}^0}(g,\ell)|^{O(1/\epsilon)}$, which is much smaller than $O((\log (n/\epsilon)/\epsilon)^{O(1/\epsilon)})$.

A detailed description of our algorithm is given in Algorithm~\ref{alg:knapsack}.



\begin{algorithm}[t!]
  \caption{$\textsc{Knapsack}$}\label{alg:knapsack}
  \begin{algorithmic}[1]
    \Require{A monotone submodular function $g:2^E \to \bbR_+$, a linear function $\ell:2^E \to \bbR$, a weight function $w :E \to \bbR_+$, and $\epsilon \in (0,1)$.}
    \Ensure{A set $S \subseteq E$ satisfying $w(S) \leq 1$.}
    \For{each choice of $v_g,v_\ell \in V_{\epsilon,n}(g,\ell)$}
      \State{$E_{\rm L} \leftarrow$ the set of large elements with respect to $v_g$ and $v_\ell$.}
      \State{$E_{\rm S} \leftarrow$ the set of small elements with respect to $v_g$ and $v_\ell$.}
      \State{$\caS \leftarrow \emptyset$.}
    \State{$M := \lfloor \frac{1}{(1-c_g)\epsilon^6}\rfloor$}
    \For{each choice of $m$ from $\set{0,1,\ldots,M}$}
      \For{each choice of $\set{\gamma^t_i},\set{\lambda_i}$ from $ V_{\epsilon,m}(g,\ell)/m$}
        \For{each choice of $\gamma^0_{\rm S}, \lambda_{\rm S}$ from $V_{\epsilon,n}(g,\ell)$}
          \For{each choice of $\set{\gamma_{\rm S}^\epsilon,\ldots,\gamma_{\rm S}^{1-\epsilon}}$ from $V_{\epsilon,n,\gamma_{\rm S}^0}(g,\ell)$}
            \State{$(\by_1,\ldots,\by_m,\bz) \leftarrow \textsc{GuessingContinuousGreedy}_{\epsilon,\epsilon}(g,\ell,w,E_{\rm L},E_{\rm S},m,\set{\gamma^t_i}, \set{\lambda_i},\set{\gamma^t_{\rm S}},\lambda_{\rm S})$.}
            \State{$S \leftarrow \textsc{Rounding}_\epsilon(w,E_{\rm L}, E_{\rm S}, m,\set{\by_i},\bz)$.}
            \State{$\caS \leftarrow \caS \cup \set{S}$.}
          \EndFor
        \EndFor
      \EndFor
    \EndFor
    \EndFor
    \State{\Return $\arg\max_{S \in \caS}g(S)+\ell(S)$.}
  \end{algorithmic}
\end{algorithm}

\begin{proof}[Proof of Theorem~\ref{the:gl}]
  Consider the case that $v_g$ and $v_\ell$ satisfy~\eqref{eq:guess-g(O)-and-l(O)}, $m = |O_{\rm L}|$, and $\set{\gamma^t_i},\set{\lambda_i}$, $\set{\gamma^t_{\rm S}}$, and $\lambda_{\rm S}$ are good guesses.
  Let $S$ be the (random) set obtained with these guesses.
  By Lemma~\ref{lem:rounding}, we have
  \begin{align}
    \E[g(S) + \ell(S)]
    \geq (1-\epsilon)(\widehat{G}(\bx) + L(\bx)) - O(\epsilon)(g(O)+\ell(O) + v_g  + v_\ell).
    \label{eq:g(S)+l(S)}
  \end{align}
  Conditioned on the event that \textsf{GuessingContinuousGreedy} succeeds,
  by (i) and (ii) of Lemma~\ref{lem:continuous-greedy}, we get
  \begin{align}
    \eqref{eq:g(S)+l(S)} & \geq (1-\epsilon)(1-1/e-O(\epsilon))g(O) + (1-\epsilon)(1-O(\epsilon))\ell(O) - O(\epsilon)(g(O)+\ell(O)+v_g+v_\ell + 8 d_{g,\ell}) \nonumber \\
    & \geq
    (1-1/e)g(O) + \ell(O) - O(\epsilon)(g(O)+\ell(O)).
    \label{eq:g(S)+l(S)-good}
  \end{align}
  Since \textsf{GuessingContinuousGreedy} succeeds with probability at least $1-\epsilon$,
  we get
  \begin{align*}
    \E[g(S) + \ell(S)]
    & \geq
    (1-\epsilon)\cdot \eqref{eq:g(S)+l(S)-good}
    \geq (1-1/e)g(O) + \ell(O) - O(\epsilon)(g(O)+\ell(O)).
  \end{align*}
  Since Algorithm~\ref{alg:knapsack} outputs the set with the maximum objective, we have the desired property on the objective value.

  It is clear that the output of Algorithm~\ref{alg:knapsack} has weight at most $1$ because \textsf{Rounding} always outputs a set of weight at most $1$.

  For arbitrary $\gamma \in V_{\epsilon,n}(g,\ell)$, the time complexity of Algorithm~\ref{alg:knapsack} is
  \begin{align*}
    & O\Bigl(\compgeneralM\Bigr) \cdot |V_{\epsilon,n}(g,\ell)|^{O(1)} \cdot |V_{\epsilon,n,\gamma}(g,\ell)|^{O(1/\epsilon)} \cdot |V_{\epsilon,m}(g,\ell)|^{O(M/\epsilon)}\\
    & =
    O\Bigl(\compgeneralM\Bigr) \cdot \Bigl(\frac{\log (n/\epsilon)}{\epsilon} \Bigr)^{O(1)}\cdot  \Bigl(\frac{1}{\epsilon}\log \frac{1}{1-c_g} \Bigr)^{O(1/\epsilon)} \cdot \Bigl(\frac{\log (M/\epsilon)}{\epsilon} \Bigr)^{O(M/\epsilon)}\\
    & =
    O\Bigl(\frac{n}{(1-c_g)^2\epsilon^{15}}\log \frac{n}{(1-c_g)\epsilon}+\frac{n^4}{\epsilon}+\frac{n^2}{\epsilon^3}\log \frac{1}{\epsilon}\Bigr) \cdot \Bigl(\frac{\log n}{\epsilon} \Bigr)^{O(1)} \cdot \\
    & \qquad \Bigl(\frac{1}{\epsilon}\log \frac{1}{1-c_g} \Bigr)^{O(1/\epsilon)}  \cdot \Bigl(\frac{1}{\epsilon}\log \frac{1}{1-c_g} \Bigr)^{O(1/((1-c_g)\epsilon^7))}\\
    & =
    \compgl.
  \end{align*}
  Hence, we have the desired time complexity.

  By replacing $\epsilon$ with $\epsilon/C$ for a large constant $C$ (to change $O(\epsilon)$ to $\epsilon$), we have the desired result.
\end{proof}


\section{The Budget Allocation Problem}\label{sec:applications}

In this section, we bound the curvature of the submodular function that represents the budget allocation problem, and we confirm that our algorithm can be applied to the budget allocation problem in order to obtain an approximation factor better than $1-1/e$.


We formally define the budget allocation problem.
The input consists of a bipartite graph with the bipartition $A \cup B$, a weight function $w:A \to [0,1]$, a capacity function $c: A \to \bbN$, and a probability function $p:A \to [0,1]$.
Intuitively speaking, the sets $A$ and $B$ correspond to media channels and customers, respectively.
Each edge $(a,b)$ in the bipartite graph represents the potential influence of media channel $a$ on customer $b$.
Consider a budget allocation $\bb \in \bbZ_+^A$ to $A$ with $\bb(a) \leq c(a)$ and $\sum_{a \in A}\bb(a) w(s) \leq 1$.
If a node $a$ is allocated a budget of $\bb(a)$, it makes $\bb(a)$ independent trials to activate each adjacent node $b$.
The probability that $b$ is activated by $a$ in each trial is $p(a)$.
Thus, the probability that $b$ becomes active is
\[
  1 - \prod_{a \in \Gamma(b)}p(a)^{\bb(a)},
\]
where $\Gamma(b)$ denotes the set of nodes in $A$  adjacent to $b$.
Hence, the expected number of activated target nodes is
\[
  \sum_{b \in B}\Bigl(1 - \prod_{a \in \Gamma(b)}p(a)^{\bb(a)}\Bigr).
\]
The objective of this problem is to find the budget allocation that maximizes the expected number of activated target nodes.

We can recast the problem using a submodular function.
For each $a \in A$, let $E_a = \set{(a,i) \mid i \in c(a)}$, and let $E = \bigcup_{a \in A}E_a$.
Then, we define $f:2^E \to \bbR_+$ as
\[
  f(S) = \sum_{b \in B}\Bigl(1 - \prod_{a \in \Gamma(b)}p(a)^{|S \cap E_a|}\Bigr).
\]
Further, we define $w':2^E \to [0,1]$ to be $w'((a,i)) = w(a)$.
Then, the budget allocation problem is equivalent to maximizing $f(S)$ subject to $w'(S) \leq 1$.

We now observe several properties of $f$.
\begin{lemma}\label{lem:ba-marginal-gain}
  Let $S \subsetneq E$ and $(a,i) \in E \setminus S$.
  Then,
  \[
    f_S((a,i)) = \sum_{b \in B: a \in \Gamma(b)}(1-p(a))\prod_{a' \in \Gamma(b)}p(a')^{|S \cap E_{a'}|}.
  \]
\end{lemma}
\begin{proof}
  For each $b \in B$, we define a function $g^b:2^E \to \bbR_+$ as $g^b(T) = 1 - \prod_{a \in \Gamma(b)}p(a)^{|T \cap E_a|}$.
  Note that $f_S((a,i)) = \sum_{b \in B}g^b_S((a,i))$.

  If $a \not \in \Gamma(b)$, then we clearly have $g^b_S((a,i)) = 0$.
  If $a \in \Gamma(b)$, then we have
  \[
    g^b_S((a,i)) = (1-p(a))\prod_{a' \in \Gamma(b)}p(a')^{|S \cap E_{a'}|}.
  \]
  Summing $g^b_S((a,i))$ over all $b \in B$, we obtain the claim.
\end{proof}

\begin{corollary}\label{cor:ba-submodular}
  The function $f$ is submodular.
\end{corollary}
\begin{proof}
  From Lemma~\ref{lem:ba-marginal-gain}, it is easy to see that $f_S((a,i)) \geq f_T((a,i))$ holds for $S \subseteq T\subsetneq E$ and $(a,i) \in E \setminus T$.
\end{proof}

\begin{corollary}\label{cor:ba-curvature}
   The curvature $c_f$ of $f$ satisfies
  \[
    c_f \leq 1- \min_{a \in A}\min_{b \in B: a \in \Gamma(b)}p(a)^{c(a)-1}\prod_{a' \in \Gamma(b) \setminus \set{a}}p(a')^{c(a')}.
  \]
\end{corollary}
\begin{proof}
  From Lemma~\ref{lem:ba-marginal-gain}, we have
  \begin{align*}
    f_{E\setminus (a,i)}((a,i)) & = \sum_{b \in B: a \in \Gamma(b)}(1-p(a))p(a)^{c(a)-1}\prod_{a' \in \Gamma(b) \setminus \set{a}}p(a')^{c(a')},   \\
    f((a,i)) & = \sum_{b \in B: a \in \Gamma(b)}(1-p(a)).
  \end{align*}
  Hence,
  \begin{align*}
    c_f
    & = 1 - \min_{(a,i) \in E}\frac{f_{E\setminus (a,i)}((a,i))}{f((a,i))}\\
    & = 1- \min_{a \in A}\frac{\sum_{b \in B: a \in \Gamma(b)}(1-p(a))p(a)^{c(a)-1}\prod_{a' \in \Gamma(b) \setminus \set{a}}p(a')^{c(a')}}{\sum_{b \in B: a \in \Gamma(b)}(1-p(a))} \\
    & \leq
    1- \min_{a \in A}\min_{b \in B: a \in \Gamma(b)}p(a)^{c(a)-1}\prod_{a' \in \Gamma(b) \setminus \set{a}}p(a')^{c(a')}.
    \qedhere
  \end{align*}
\end{proof}

From our main result (Theorem~\ref{the:intro}) and Corollaries~\ref{cor:ba-submodular} and~\ref{cor:ba-curvature}, when the capacity of each node $a \in A$ is bounded by a constant and the number of vertices adjacent to each node $b\in B$ is bounded by a constant, we obtain a polynomial-time algorithm whose approximation ratio is strictly better than $1-1/e$.

\section*{Acknowledgments}
We thank Takanori Maehara for providing us with the problem and insightful comments.

\bibliographystyle{abbrv}
\bibliography{main}


\end{document}